\documentclass{birkjour}
\usepackage{graphicx} 
\usepackage{amssymb}
\usepackage{amsmath}
\usepackage{amsthm}
\usepackage{tikz}
\usepackage{tikz-cd}
\usetikzlibrary{%
  matrix,%
  calc,%
  arrows%
}

\usepackage{slashbox}

\newtheorem{thm}{Theorem}
\newtheorem{prop}{Proposition}
\newtheorem{lem}{Lemma}

\newcommand{\DD}{\mathcal{D}\mkern-12mu\mathcal{D}}

\theoremstyle{remark}
\newtheorem{remark}{Remark}

\title{Algebraic Limits of Sandpiles}
\author{Mikhail Shkolnikov}
\date{}

\address{
Institute of Mathematics and Informatics, Bulgarian Academy of Sciences}
\email{m.shkolnikov@math.bas.bg}

\begin{document}

\begin{abstract}
    The paper contributes to building algebraic foundations of self-organized criticality answering a previously unsolved question about the limiting structure of the extended sandpile group as well as relating it to another limit at the level of classical sandpile groups with respect to certain monomorphisms, and puts forward a concept of canonical sandpile epimorphisms, drawing unexpected consequences about the divisibility properties of the numbers of spanning trees for grid foldings.
\end{abstract}

\maketitle

\section{Introduction}
What we now call ``sandpile model'' has appeared multiple times in diverse guises throughout mathematics and physics, but has acquired a tremendous attention of a large scientific community only when it was proposed as a prototype for self-organized criticality \cite{bak1987self}. The name of the model used in the present note was coined by Deepak Dhar with an additional adjective ``Abelian'' referring in part to his discovery \cite{dhar1990self} of the fact that the set of recurrent states of the model on a finite domain has a natural structure of an Abelian group. 

One characteristic feature of self-organized critical systems is their scale transcendence, which manifests in the form of various scaling limit theorems such as \cite{pegden2013convergence} or \cite{kalinin2016tropical}. These type of results may be described as geometric or ``visual'' approach to the limiting structure of sandpiles, since the graphical representation of states is under consideration. Alternatively, one may focus on probabilistic properties of the sandpile model in the scaling limit (see for instance \cite{bhupatiraju2017inequalities}) ignoring altogether how the states look like. 

What is going to be discussed next is similar to the latter spirit, but the attention will be turned to the algebraic aspects of the theory, i.e. our goal is to prepare the ground for the abstract understanding of sandpile groups in the limit, ignoring for the moment the intricate structure and stunning beauty of the underlying recurrent configurations. Perhaps, over time, these three approaches -- geometric, probabilistic and algebraic -- will be unified.

\section{Sandpile groups and harmonic functions}
Consider a finite connected undirected graph without loops with a distinguished vertex that we call the sink. Denote by $\Gamma$ the set of non-sink vertices of this graph. A state of the sandpile model is a non-negative integer-valued function on $\Gamma.$ One may think that a value of a state at a vertex represents the number of sand grains. 

The (reduced) Laplacian $\Delta\colon\mathbb{Z}^\Gamma\rightarrow\mathbb{Z}^\Gamma$ is given by the difference of the symmetric adjacency matrix and the diagonal degree matrix with removed row and column corresponding to the sink. Note that the graph is restored from its Laplacian, the absolute value of its determinant is equal to the number of spanning trees on the graph by Kirchhoff's theorem. 

A state $\phi$ is called stable if $\phi(v)$ is less than the degree of $v$ for all $v\in\Gamma.$ If a state is unstable there exist a vertex $v$ where $\phi(v)$ is greater or equal to the degree of $v$ and a toppling (i.e. redistribution of sand grains) $\phi\mapsto\phi+\Delta\delta_v$ can be performed. Here $\delta_v$ is a function on $\Gamma$ equal to $1$ at $v$ and vanishing otherwise. A relaxation is a process of applying topplings until the state becomes stable. The process is always finite and the resulting state $\phi^\circ$ doesn't depend on the order of topplings. 

A stable state is called recurrent if it can be obtained from any other state by adding sand grains and relaxing. The sandpile group $G(\Gamma)$ consists of recurrent states with an operation given by $(\phi,\psi)\mapsto(\phi+\psi)^\circ.$ For a function $f\in\mathbb{Z}^\Gamma$ there exist a unique recurrent state $\phi$ and a function $H\in\mathbb{Z}^\Gamma$ such that $f=\phi+\Delta H.$ This gives an isomorphism $G(\Gamma)\cong\operatorname{Coker}(\Delta)=\mathbb{Z}^\Gamma\slash \Delta(\mathbb{Z}^\Gamma)$ which can be written equivalently in a form of a short exact sequence 
$$0\rightarrow\mathbb{Z}^\Gamma\xrightarrow{\Delta}\mathbb{Z}^\Gamma\rightarrow G(\Gamma)\rightarrow 0,$$ where the notaion means that the left homomorphism is injective (which follows from the non-vanishing of the determinant of $\Delta$), the right one is surjective, and the image and kernel coincide in the middle.

Denote by $\partial\Gamma\subset\Gamma$ the set of vertices adjacent to the sink. We call $\partial\Gamma$ the boundary and $\Gamma^\circ=\Gamma\backslash\partial\Gamma$ the interior of the graph. The extended sandpile group $\widetilde G(\Gamma)$ is derived by allowing the recurrent states to take real non-negative values along $\partial\Gamma$. In other words, there is a short exact sequence $$0\rightarrow\mathbb{Z}^\Gamma\xrightarrow{\widetilde\Delta}\mathbb{Z}^{\Gamma^\circ}\mkern-9mu\oplus\mathbb{R}^{\partial\Gamma}\rightarrow\widetilde G(\Gamma)\rightarrow 0,$$
where $\widetilde\Delta$ is the composition of $\Delta$ with the inclusion of $\mathbb{Z}^\Gamma=\mathbb{Z}^{\Gamma^\circ}\mkern-9mu\oplus\mathbb{Z}^{\partial\Gamma}$ to $\mathbb{Z}^{\Gamma^\circ}\mkern-9mu\oplus\mathbb{R}^{\partial\Gamma}.$ There is another short exact sequence $$0\rightarrow G(\Gamma)\rightarrow\widetilde G(\Gamma)\xrightarrow{\slash\mathbb{Z}} (\mathbb{R}\slash\mathbb{Z})^{\partial\Gamma}\rightarrow 0,$$ where $\slash\mathbb{Z}$ takes the fractional part of the extended recurrent state and $\mathbb{R}\slash\mathbb{Z}$ is the quotient of $\mathbb{R}$ by the additive subgroup $\mathbb{Z}$ which topologically is a circle.  

Let $\Delta^\circ\colon\mathbb{Z}^\Gamma\rightarrow\mathbb{Z}^{\Gamma^\circ}$ be a map given by the composition of $\Delta$ with the restriction to $\Gamma^\circ.$ For an Abelian group $A$ define the group $\mathcal{H}_A(\Gamma)$ of $A$-valued harmonic functions as the kernel of $\Delta^\circ\otimes A.$

Consider a category $\DD$ of connected graphs without loops (possibly infinite, with or without sink) represented by pairs $(\Gamma,\Delta)$ with morphisms given by embeddings $\iota\colon\Gamma_1\rightarrow\Gamma_2$ such that $\Delta_1$ is obtained from $\Delta_2$ by removing rows and columns corresponding to vertices in $\Gamma_2\backslash\iota(\Gamma_1).$ For example, the casual map $\Omega\subset\mathbb{R}^n\mapsto\Omega\cap\mathbb{Z}^n$ is upgraded to a contravariant functor from the category of open sets in $\mathbb{R}^n$ with inclusions as morphisms to $\DD.$ Note that $A^\Gamma$ and $A^{\Gamma^\circ}$ are contravariant functors from $\DD$ to Abelian groups and $\Delta^\circ\otimes A$ is their natural transformation. Therefore, $\mathcal{H}_A(\Gamma)$ is also a contravariant functor, with the induced homomorphisms being simple restrictions.

\begin{prop} The groups $\widetilde G(\Gamma)$ and $\mathcal{H}_{\mathbb{R}\slash\mathbb{Z}}(\Gamma)$ are canonically isomorphic.
\label{prop_isoextsand}
\end{prop}
\begin{proof}
The isomorphism is given by the connecting homomorphism arising in the following employment of the snake lemma where for the middle row we use the invertibility of $\Delta\otimes\mathbb{R}:$
\begin{figure}[!h]
\hspace*{-0.8cm}  
\begin{tikzpicture}
\matrix[matrix of math nodes,column sep={75pt,between origins},row
sep={30pt,between origins}](m)
{&  &|[name=m12]| 0 &|[name=m13]| 0 &  &  \\
&|[name=m21]| 0 &|[name=m22]| \mathbb{Z}^\Gamma &|[name=m23]|  \mathbb{Z}^{\Gamma^\circ}\mkern-9mu\oplus\mathbb{R}^{\partial\Gamma} &|[name=m24]| \widetilde G(\Gamma) \\
&|[name=m31]| 0 &|[name=m32]| \mathbb{R}^\Gamma &|[name=m33]| \mathbb{R}^\Gamma &|[name=m34]| 0\\
&|[name=m41]| \mathcal{H}_{\mathbb{R}\slash\mathbb{Z}}(\Gamma) &|[name=m42]| (\mathbb{R}\slash\mathbb{Z})^\Gamma &|[name=m43]| (\mathbb{R}\slash\mathbb{Z})^{\Gamma^\circ}&|[name=m44]|0\\
& &|[name=m52]| 0 &|[name=m53]| 0& \\
};
\draw[->] 
(m12) edge (m22)
(m22) edge (m32)
(m32) edge (m42)
(m42) edge (m52)

(m13) edge (m23)
(m23) edge (m33)
(m33) edge (m43)
(m43) edge (m53)

(m21) edge (m22)
(m22) edge  node[auto] {$\scriptstyle\widetilde\Delta$} (m23)
(m23) edge (m24)

(m31) edge (m32)
(m32) edge  node[auto] {$\scriptstyle\Delta\otimes\mathbb{R}$} (m33)
(m33) edge (m34)

(m41) edge (m42)
(m42) edge  node[auto] {$\scriptstyle\Delta^\circ\otimes\mathbb{R}\slash\mathbb{Z}$} (m43)
(m43) edge (m44)
;

\draw[->]
(m31) edge[dashed] (m41)
(m24) edge[dashed] (m34)
;

\draw[->]
(m41) edge[dashed] (m24)
;

\end{tikzpicture}

\end{figure}

\end{proof}

\begin{remark} In a similar way, one obtains $G(\Gamma)\cong\operatorname{Ker}(\Delta\otimes\mathbb{R}\slash\mathbb{Z}),$ the latter may be thought of as the group of ``strictly harmonic'' circle-valued functions on the graph, i.e. the vanishing of the Laplacian is required both in the interior and on the boundary.
\label{rem_isosp}
\end{remark}

Another application of the snake lemma gives the following.
\begin{prop}The groups of integer-, real- and circle- valued harmonic functions and the group $\operatorname{Coker}\Delta^\circ$ are linked through an exact sequence
$$0\rightarrow\mathcal{H}_{\mathbb{Z}}(\Gamma)\rightarrow\mathcal{H}_{\mathbb{R}}(\Gamma)\rightarrow\mathcal{H}_{\mathbb{R}\slash\mathbb{Z}}(\Gamma)\rightarrow\operatorname{Coker}\Delta^\circ\rightarrow 0.$$
\label{prop_integerrealcirc}
\end{prop}
Lie groups $\mathcal{H}_{\mathbb{R}}(\Gamma)$ and $\mathcal{H}_{\mathbb{R}\slash\mathbb{Z}}(\Gamma)$ both have dimension equal to $|\partial\Gamma|$ and $\operatorname{Coker}\Delta^\circ$ is discrete. Therefore, since $\mathcal{H}_{\mathbb{R}}(\Gamma)$ is connected (being an $\mathbb{R}$-vector space), the image of $\mathcal{H}_{\mathbb{R}}(\Gamma)$ in $\mathcal{H}_{\mathbb{R}\slash\mathbb{Z}}(\Gamma)$ is identified with the maximal connected subgroup of $\widetilde G(\Gamma)$ that we denote by $\widetilde G_0(\Gamma).$ We observe the isomorphisms $\widetilde G_0(\Gamma)\cong \mathcal{H}_{\mathbb{R}}(\Gamma)\slash\mathcal{H}_{\mathbb{Z}}(\Gamma)$ and $\operatorname{Coker}\Delta^\circ\cong\widetilde G(\Gamma)\slash\widetilde G_0(\Gamma).$ 

Yet another snake lemma diagram (which is also omitted) explains the relation between the cokernel of $\Delta^\circ$ and the usual sandpile group.
\begin{prop}
There is an exact sequence 
$$0\rightarrow\mathcal{H}_{\mathbb{Z}}(\Gamma)\rightarrow\mathbb{Z}^{\partial\Gamma}\rightarrow G(\Gamma)\rightarrow\operatorname{Coker}\Delta^\circ\rightarrow 0.$$
\label{prop_spcokerdcirc}
\end{prop}
Here, the map $\mathcal{H}_{\mathbb{Z}}(\Gamma)\rightarrow\mathbb{Z}^{\partial\Gamma}$ is the composition of $\Delta$ and the projection to $\partial\Gamma.$ Denote by $G_0(\Gamma)$ the image of the homomorphism $\mathbb{Z}^{\partial\Gamma}\rightarrow G(\Gamma).$ It is a subroup of the sandpile group generated by adding sand at the boundary. Note that $G_0(\Gamma)=\widetilde G_0(\Gamma)\cap G(\Gamma)$ and $\operatorname{Coker}\Delta^\circ\cong G(\Gamma)\slash G_0(\Gamma).$ This implies the following lemma.
\begin{lem} If $G(\Gamma)$ is generated by adding sand at the boundary, then $$\widetilde G(\Gamma)\cong\mathcal{H}_{\mathbb{R}}(\Gamma)\slash\mathcal{H}_{\mathbb{Z}}(\Gamma).$$
\label{lem_isoespcirch}
\end{lem}
For the rest of this section we will be concerned with the inverse limit $$\lim_{\mathbb{Z}^2\leftarrow\Gamma}\widetilde G(\Gamma)$$ over all finite $\Gamma\subset\mathbb{Z}^2.$ It can be replaced with a limit over any sequence of finite subgraphs $\Gamma_1\subset\Gamma_2\subset\dots\subset\mathbb{Z}^2$ such that $\cup_{n=1}^\infty\Gamma_n=\mathbb{Z}^2.$ We prefer to assume that all $\Gamma_n$ are convex (i.e. intersections of convex domains with the lattice) for the following reason. 
\begin{lem} For convex $\Gamma_1\subset\Gamma_2\subset\mathbb{Z}^2$ the restriction homomorphism $$\mathcal{H}_A(\Gamma_2)\rightarrow\mathcal{H}_A(\Gamma_1)$$ is surjective.
\label{lem_convsurj}
\end{lem}
\begin{proof}
A {\it diamond} is an intersection of $\mathbb{Z}^2$ with a  rotated by $45$ degrees rectangle $[a,b]\times[c,d]\subset\mathbb{R}^2$ for some real $a,b,c,d$. Diamonds, among finite convex discrete domains with non-empty interior on the square lattice, are characterized by the property that each vertex of their boundary has more than one neighbor in the complement of the domain. Note that an intersection of two diamonds is again a diamond.

A diamond hull $DH(\Gamma)$ of a finite $\Gamma\subset\mathbb{Z}^2$ is the minimal by inclusion diamond containing $\Gamma.$ As the first step, we observe that for convex $\Gamma,$ the restriction homomorphism $\mathcal{H}_A(DH(\Gamma))\rightarrow\mathcal{H}_A(\Gamma)$ is an isomorphism. This follows from the fact that convexity of $\Gamma$ implies that there is no vertex $v$ in the complement of $\Gamma$ such that it is adjacent to two distinct vertices $w_1,w_2\in\partial\Gamma$ and both $w_1$ and $w_2$ have only $v$ as an adjacent vertex in the complement. Thus, if $\Gamma$ is not a diamond and connected, it has a vertex $w$ on the boundary adjacent to a unique vertex $v$ in the complement. Consider a new domain $\Gamma'=\Gamma\cup\{v\}$ for which $w$ is not a boundary vertex. Then, for $\phi\in\mathcal{H}_A(\Gamma),$ there is a unique $\phi'\in\mathcal{H}_A(\Gamma')$ extending $\phi$ to $v$ by $-(\Delta\phi)(w),$ i.e. the restriction map $\mathcal{H}_A(\Gamma')\rightarrow\mathcal{H}_A(\Gamma)$ is an isomorphism. Note also that $DH(\Gamma)=DH(\Gamma')$ Now, we may consider a finite chain of inclusions $$\Gamma\subset\Gamma'\subset\Gamma''\subset\dots\subset\Gamma^{(n)}\subset\dots\subset DH(\Gamma),$$ where $\Gamma^{(n+1)}$ is obtained by adding one vertex to $\Gamma^{(n)}$ in the same way as we constructed $\Gamma'$ from $\Gamma,$ thus each inclusion induces and isomorphism at the level of $A$-valued harmonic functions.

Now we may replace $\Gamma_1$ and $\Gamma_2$ with their diamond hulls $D_1$ and $D_2$. Assuming that the inclusion $D_1\subset D_2$ is strict, there is a vertex $v\in D_2\backslash D_1$ adjacent to a vertex $w$ in $\partial D_1.$ Consider a domain $D_1'=D_1\cup\{v\},$ note that $\partial D_1'=\partial D_1\cup\{v\}.$  Therefore, since we are not requiring the vanishing of the Laplacian on the boundary, an element $\phi\in\mathcal{H}_A(D_1)$ can be extended arbitrarily to $v,$ i.e. the restriction homomorphism $\mathcal{H}_A(D_1')\rightarrow\mathcal{H}_A(D_1)$ is surjective (but not injective for $A\neq 0$). By the previous paragraph, we may replace $D_1'$ by its diamond hull, resulting in a strictly bigger diamond still contained in $D_2.$ Thus, we decompose the restriction homomorphism from $D_2$ to $D_1$ into a finite chain of restrictions, each of which is surjective. \end{proof}
\begin{lem}
For a finite $\Gamma\subset\mathbb{Z}^2$ the sandpile group $G(\Gamma)$ is generated by adding sand at the boundary.
\label{lem_genbound}
\end{lem}

\begin{proof} By induction one can show that adding a grain of sand at a vertex can be achieved by adding grains to the left of it and at the boundary. 
\end{proof}

\begin{remark}
    For abstract graphs, which are not subgraphs of the lattice, the above statement is not necessarily true. One of the simplest examples is provided by a graph with three non-sink vertices $v_-,$ $v_0$ and $v_+$ such that each of the two are connected by a single edge and only $v_0$ is connected to the sink (by a single edge), thus $v_0$ is the only boundary vertex of this graph. The number of spanning trees on it is $3,$ therefore the sandpile group is isomorphic to $\mathbb{Z}\slash 3\mathbb{Z},$ which has a single non-trivial automorphism given by multiplication by $-1,$ realized by the graph automorphism permuting $v_-$ and $v_+.$ The neutral element of this sandpile group is invariant under the permutation, and if we add a single grain at the boundary vertex $v_0,$ the state will remain symmetric after a relaxation, i.e. its result is again the neutral element. We conclude that the subgroup of the sandpile group generated by adding sand at the boundary is trivial in this example.
\end{remark}

\begin{lem} For a sequence of Abelian groups and epimorphisms 
$$A_1\xleftarrow{\pi_1}A_2\xleftarrow{\pi_2}A_3\xleftarrow{\pi_3}\cdots$$ there is a short exact sequence 
$$0\rightarrow\lim_{\leftarrow} A_n\rightarrow\prod_{n=1}^\infty A_n\xrightarrow{\Psi}\prod_{n=1}^\infty A_n\rightarrow 0,$$ where $\displaystyle\lim_{\leftarrow} A_n$ is the inverse limit, $\prod_{n=1}^\infty A_n$ is the direct product and \begin{equation}\label{eq_psi}\tag{$\star$}
    \Psi(a_1,a_2,\dots)=(\pi_1(a_2)-a_1,\pi_2(a_3)-a_2,\dots).\end{equation}
\label{lem_colimit}
\end{lem}
\begin{proof} 
The inverse limit is by definition the kernel of $\Psi.$ The surjectivity is implied by surjectivity of $\pi_1,\pi_2 \dots$ 
\end{proof}

For an arbitrary finite $\Gamma,$ there is a natural homomorphism from $\mathcal{H}_{\mathbb{R}}(\Gamma)$ to $\widetilde{G}(\Gamma)$ given by taking $\Delta$ of a real-valued harmonic function, which results in a real-valued function supported at the boundary of $\Gamma,$ and taking it modulo the image of $\widetilde{\Delta}$ described at the beginning of this section. Considering the colimit for $\Gamma\rightarrow\mathbb{Z}^2,$ we still have a homomorphism from $\mathcal{H}_{\mathbb{R}}(\mathbb{Z}^2)$ to $\displaystyle\lim_{\mathbb{Z}^2\leftarrow\Gamma}\widetilde G(\Gamma).$ It is easy to see that its kernel is $\mathcal{H}_{\mathbb{Z}}(\mathbb{Z}^2)$ (this may be done by applying the interpretation of the extended sandpile group as the cockernel of $\widetilde{\Delta}$ and the injectivity of $\Delta$ on finite graphs),i.e. we have an exact sequence $$0\rightarrow\mathcal{H}_{\mathbb{Z}}(\mathbb{Z}^2)\rightarrow\mathcal{H}_{\mathbb{R}}(\mathbb{Z}^2)\rightarrow\displaystyle\lim_{\mathbb{Z}^2\leftarrow\Gamma}\widetilde G(\Gamma)\rightarrow X\rightarrow 0.$$
In \cite{lang2019harmonic}, it was asked: is it true that $X=0$? 

\begin{thm}[Affirmative answer to Question 1 from \cite{lang2019harmonic}]The natural inclusion of
$\mathcal{H}_{\mathbb{R}}(\mathbb{Z}^2)\slash\mathcal{H}_{\mathbb{Z}}(\mathbb{Z}^2)$ to $\displaystyle\lim_{\mathbb{Z}^2\leftarrow\Gamma}\widetilde G(\Gamma)$ is an isomorphism.
\label{thm_answer}
\end{thm}
\begin{proof}
Consider an ascending chain of convex discrete domains ${\Gamma_n}$ converging to $\mathbb{Z}^2.$ For the fifth and last time we will apply the snake lemma starting with the following diagram, where the rows are short exact sequences (this follows from applying Lemma \ref{lem_genbound} and then Proposition \ref{prop_isoextsand} followed by Lemma \ref{lem_isoespcirch}), and the two squares are commutative (which is the compatibility of $\Psi_A$ with the inclusion and projection at the level of coefficients):

\[ \begin{tikzcd}
 0\arrow{r}&\prod\mathcal{H}_{\mathbb{Z}}(\Gamma_n) \arrow{r} \arrow{d}{\Psi_{\mathbb{Z}}} &  \prod\mathcal{H}_{\mathbb{R}}(\Gamma_n) \arrow{d}{\Psi_{\mathbb{R}}}\arrow{r}& \prod\mathcal{H}_{\mathbb{R}\slash\mathbb{Z}}(\Gamma_n)\arrow{r}\arrow{d}{\Psi_{\mathbb{R}\slash\mathbb{Z}}}&0 \\%
 0\arrow{r}&\prod\mathcal{H}_{\mathbb{Z}}(\Gamma_n) \arrow{r}&  \prod\mathcal{H}_{\mathbb{R}}(\Gamma_n)\arrow{r}&\prod\mathcal{H}_{\mathbb{R}\slash\mathbb{Z}}(\Gamma_n)\arrow{r}&0.
\end{tikzcd}
\]

Here, each downward arrow $\Psi_A,$ for $A=\mathbb{Z}$, $\mathbb{R}$ or $\mathbb{R}\slash\mathbb{Z},$ stands for the homomorphism given by a formula analogous to (\ref{eq_psi}) with $\pi_n$ being replaced by $\pi^A_n,$ the restriction homomorphism $\mathcal{H}_A(\Gamma_{n+1})\rightarrow\mathcal{H}_A(\Gamma_{n}).$ The snake lemma for the above diagram gives an exact sequence
$$\operatorname{Ker}\Psi_{\mathbb{Z}}\rightarrow\operatorname{Ker}\Psi_{\mathbb{R}}\rightarrow\operatorname{Ker}\Psi_{\mathbb{R}\slash\mathbb{Z}}\rightarrow\operatorname{Coker}\Psi_{\mathbb{Z}}\rightarrow\operatorname{Coker}\Psi_{\mathbb{R}}\rightarrow\operatorname{Coker}\Psi_{\mathbb{R}\slash\mathbb{Z}}.$$

Note that $\operatorname{Ker}\Psi_A$ is by definition the colimit of $\mathcal{H}_A(\Gamma_n)$ which is trivially equal to $\mathcal{H}_A(\mathbb{Z}^2).$ By Lemma \ref{lem_convsurj}, all the homomorphisms $\pi_n^A$ are surjective, thus we may apply Lemma \ref{lem_colimit}, i.e. all the cokernels in the above six-term sequence vanish, which results in $X=0.$ In other words, $\mathcal{H}_{\mathbb{R}}(\mathbb{Z}^2)\slash\mathcal{H}_{\mathbb{Z}}(\mathbb{Z}^2)$ and $\displaystyle\lim_{\mathbb{Z}^2\leftarrow\Gamma}\widetilde G(\Gamma)$ are naturally isomorphic.
\end{proof}

\section{Monomorphisms and epimorphisms}
For an Abelian group $A$ denote by $A^*$ its Pontryagin dual, i.e. $$A^*=\operatorname{Hom}(A,\mathbb{R}\slash\mathbb{Z}).$$ In general, if $A$ is finite then it is isomorphic to $A^*.$ However, a choice of such an isomorphism is equivalent to an extra structure on $A$ which is a non-degenerate paring on $A$ with values in the circle $\mathbb{R}\slash\mathbb{Z}.$ For a sandpile group of a graph such a paring comes naturally. 
\begin{prop} Sandpile groups of finite graphs are canonically self-dual.
\label{prop_selfdual}
\end{prop}
\begin{proof}
To construct the pairing on $G(\Gamma)$ we apply natural isomorphisms $G(\Gamma)\cong\operatorname{Coker}\Delta$ and $G(\Gamma)\cong\operatorname{Ker}(\Delta\otimes\mathbb{R}\slash\mathbb{Z})$ to the components of the pairing $$Q\colon\operatorname{Coker}\Delta\times\operatorname{Ker}(\Delta\otimes\mathbb{R}\slash\mathbb{Z})\rightarrow \mathbb{R}\slash\mathbb{Z}$$ given by $$Q(\phi+\Delta\mathbb{Z}^\Gamma,\psi)=\sum_{v\in\Gamma}\phi(v)\psi(v).$$ 

First, we need to show that $Q$ is well defined, i.e. the right hand side doesn't depend on the choice of a representative $\phi$ in the coset. This follows from observing $$\sum_{v\in\Gamma}(\Delta\delta_w)(v)\psi(v)=(\Delta\psi)(w)=0$$ for any $w\in\Gamma.$ Second, we need to show that the homomorphism $$\psi\in\operatorname{Ker}(\Delta\otimes\mathbb{R}\slash\mathbb{Z})\mapsto Q(-,\psi)\in(\operatorname{Coker}\Delta)^*$$ is injective (and therefore bijective due to the same cardinality). Indeed, its kernel is trivial since for  non-zero $\psi$ there exist $w\in\Gamma$ such that $\psi(w)\neq 0$ and so $Q(\delta_w+\Delta\mathbb{Z}^\Gamma,\psi)=\psi(w)\neq 0.$
\end{proof}

This proposition (and the vanishing of $\operatorname{Ext}^1(A,\mathbb{R}\slash\mathbb{Z})$ for a finite Abelian group $A$) implies that there is a canonical one-to-one correspondence between monomorphisms from $G(\Gamma_1)$ to $G(\Gamma_2)$ and epimorphisms from $G(\Gamma_2)$ to $G(\Gamma_1)$ for any two finite graphs.  

A particular set of monomorphisms between some sandpile groups was introduced in \cite{lang2022sandpile} in the case of certain types of tilings. The construction presented there is quite involved, thus it seems to be appropriate to write below a more direct definition of these monomorphisms in the simplest case of a rectangular tiling. 

Let $\Gamma(p,q)$ be the portion of the square lattice given by the intersection with an open rectangle $(0,p)\times(0,q),$ where $p$ and $q$ are positive integers. Think of elements in $G(\Gamma(p,q))$ as of strictly harmonic circle valued functions on $\Gamma(p,q)$ as in Remark \ref{rem_isosp}. Then, for any pair $(m,n)$ of positive integers, the monomorphism from $G(\Gamma(p,q))$ to $G(\Gamma(mp,nq))$ is given by assigning to $\phi$ $$(x,y)\in\Gamma(mp,nq)\mapsto(-1)^{\lfloor\frac{x}{p}\rfloor+\lfloor\frac{y}{q}\rfloor}\phi(\mu^{-1}_{p,\lfloor\frac{x}{p}\rfloor}(x),\mu^{-1}_{q,\lfloor\frac{y}{q}\rfloor}(y)),$$ where $\lfloor\cdot\rfloor$ denotes the integer part, $\phi$ is extended by $0$ outside $\Gamma(p,q),$ and $\mu_{p,k}\colon [0,p]\rightarrow[kp,(k+1)p]$ is given by $$\mu_{p,k}(x)=\begin{cases}
			kp+x, & \text{if $k$ even}\\
            (k+1)p-x, & \text{otherwise.}
		 \end{cases}$$ The dual map (which by above mentioned considerations is an epimorphism) is given by assigning to an element $\psi$ of $G(\Gamma(mp,nq))$ a circle-valued function $$(x,y)\in\Gamma(p,q)\mapsto\sum_{k=0}^{m-1}\sum_{l=0}^{n-1}(-1)^{k+l}\psi(\mu_{p,k}(x),\mu_{q,l}(y)).$$ 

Next we observe that the composition $$G(\Gamma(p,q))\hookrightarrow G(\Gamma(mp,nq))\twoheadrightarrow G(\Gamma(p,q))$$ of the epimorphism and the monomorphism is the identity homomorphism multiplied by $mn$. There are two extreme cases: when $mn$ is coprime with $g(p,q)=|G(\Gamma(p,q))|,$ in which case $G(\Gamma(p,q))$ is realized as a direct summand of $G(\Gamma(mp,nq)),$ or when $mn$ is divisible by $g(p,q),$ and therefore we realize the group $G(\Gamma(p,q))$ as a subgroup of the kernel of the epimorphism, which implies that $(g(p,q))^2$ divides $g(mp,nq).$ 

\begin{table}
    
\begin{centering}
\begin{tabular}{|l||*{6}{c|}}\hline
\backslashbox{$m$}{\vspace{-5pt}$n$}
&\makebox{2}&\makebox{3}&\makebox{4}&\makebox{5}&\makebox{6}\\ \hline\hline
2 & 4 &15&56&209&780\\\hline
3 &15&192&2415&30305&380160\\\hline
4 &56&2415&100352&4140081&170537640\\\hline
5 &209&30305&4140081&557568000&74795194705\\\hline
6 &780&380160&170537640&74795194705&32565539635200\\\hline
\end{tabular}

\end{centering}

\vspace{15pt}
\caption{Orders of sandpile groups on $(0,m)\times(0,n)\cap\mathbb{Z}^2.$ Observe the divisibility properties in accordance with Theorem 2. For instance, $g(2,4)$ is divisible by $g(2,2),$ but not by $(g(2,2))^2$ since the scaling factor $2=1\times 2$ is not $0$ modulo $g(2,2)=4;$ on the other hand $(g(2,2))^2$ divides $g(4,4).$}

\end{table}

A simpler divisibility property which is implied by the existence of epi-/mono-morphism is that $g(p,q)$ always divides $g(mp,nq).$ As it was pointed out to the author, this fact was observed before. The only reference, however, appears to be a webpage \cite{oeis} in the Online Encyclopedia of Integer Sequences (OEIS), where one may find a contribution by Peter Bala made in 2014, claiming that this property follows from some explicit formula.

The relationship between two rectangles $\Gamma(p,q)$ and $\Gamma(mp,nq)$ represented by a continuous folding map $\Phi_{p,q,m,n}:[0,mp]\times [0,nq]\rightarrow [0,p]\times[0,q],$ which allows to define the corresponding morphisms between their sandpile groups, is a very particular case of a much more general situation. Intuitively, one should imagine that an $mp\times nq$ rectangular piece of grid paper is folded to a $p\times q$ rectangular piece. In this process, every fold is made across a horizontal or vertical line passing through the nodes of the lattice, and the graph structure of the square lattice is preserved. If we cut the bigger piece along the folds, we would obtain $mn$ copies of the $p\times q$ rectangular piece, each of which is placed isometrically on $(0,p)\times(0,q])$ either preserving or reversing the orientation. The defined above sandpile morphisms correspond to a matrix with rows and columns indexed by $\Gamma(p,q)$ and $\Gamma(mp,nq),$ with entries being $0$, $1$ or $-1,$ with $0$ corresponding to either a point of $\Gamma(mp,nq)$ being on a fold (and any point of $\Gamma(p,q)$), or two points $v\in\Gamma(p,q)$ and $w\in\Gamma(mp,nq)$ such that $v\neq\Phi_{p,q,m,n}(w)$, and $\pm 1$ corresponding to the case of $v=\Phi_{p,q,m,n}(w)$ with the sign depending on whether the orientation is preserved or reversed.

In fact, one is allowed to perform the folds not only along the vertical and horizontal lines, but also along the lines of slope $\pm 1.$ This type of lines is distinguished among all other lines since the corresponding reflections preserve the graph structure of $\mathbb{Z}^2.$ With this in mind, we give a formal (and, perhaps, not optimal) definition of a {\it grid folding} with an adjective ``grid'' referring to grid paper, some of its instances are shown on Figure \ref{fig_folding}. This concept is a vast generalization of a $DC$-tiling introduced previously in \cite{lang2022sandpile}.

Let $\Omega, \Omega'\subset\mathbb{R}^2$ be two bounded connected domains such that each is equal to the closure of its interior, and let $\Phi\colon\Omega'\rightarrow\Omega$ be a continuous map. We call $\Phi$ a grid folding if the following conditions are satisfied:

\begin{itemize}
    \item the preimage of the interior of $\Omega$ has $d_\Phi$ connected components such that the restriction of $\Phi$ to each is an isometry preserving the lattice;
    \item the intersection of the preimage of the boundary of $\Omega$ with the interior of $\Omega'$ is a union of segments each of which extends to a line passing through some vertex of the lattice and having a slope equal to either $0,\infty$ or $\pm 1;$
    \item for a pair of connected components $\Omega_1$ and $\Omega_2$ in $\Phi^{-1}(\Omega^\circ)$ such that the intersection of their closures is a set with more than one point, this intersection is contained in one of the above mentioned lines and the map $(\Phi|_{\Omega_2})^{-1}\circ \Phi|_{\Omega_1}$ extends to the reflection with respect to this line.  
\end{itemize}

We call $d_\Phi$ the {\it degree} of the folding $\Phi.$ For instance, for the case of the folding $\Phi_{p,q,m,n}:[0,mp]\times [0,nq]\rightarrow [0,p]\times[0,q],$ its degree is $mn.$ To a folding $\Phi$ one associates the matrix $M_\Phi$ indexed by $\Gamma=\Omega^\circ\cap\mathbb{Z}^2$ and $\Gamma'=(\Omega')^\circ\cap\mathbb{Z}^2,$  with entries being $0$ and $\pm 1,$ depending on wether a point is on the fold, and whether the orientation is preserved or reversed, otherwise. This matrix and its transpose induce a monomorphism and an epimorphism between $G(\Gamma)$ and $G(\Gamma').$ Their composition is equal to the multiplication by $d_\Phi$ homomorphism (this is analogous to the composition of cohomology pullback and push-forward associated with a topological covering map). This homomorphism is zero if the degree is divisible by the cardinality of $G(\Gamma)$, which gives the second assertion of the following theorem.

\begin{thm}
    Let $\Gamma$ be a degree $d$ grid folding of $\Gamma'.$  Then, the number of spanning trees $g(\Gamma)$ on $\Gamma$ divides the number of spanning trees $g(\Gamma')$ on $\Gamma'.$  Moreover, if $d$ is divisible by $g(\Gamma)$, then $(g(\Gamma))^2$ divides $g(\Gamma').$
    \label{thm_divisibility}
\end{thm}

Now we would like to describe the relationship between the sandpile monomorphisms and the restriction homomorphisms between the corresponding extended sandpile groups. Namely, the monomorphism may be seen as a section of the restriction defined over the usual sandpile group.

\begin{prop}
For two domains $\Gamma=\Omega^\circ\cap\mathbb{Z}^2$ and $\Gamma'=(\Omega')^\circ\cap\mathbb{Z}^2$ related by a grid folding $\Phi:\Omega'\rightarrow\Omega$ such that $\Omega\subset\Omega'$ and $\Phi$ restricted to $\Omega$ is the identity map, there is the following commutative diagram: $$\begin{tikzcd}
G(\Gamma) \arrow[hook, r] \arrow[d, hook]
& G(\Gamma') \arrow[d, hook] \\
\widetilde{G}(\Gamma)
& \arrow[l, two heads] \widetilde{G}(\Gamma').
\end{tikzcd} $$
\label{prop_section}
\end{prop}
A particular corollary of this proposition (the proof of which is  trivial) is that the direct limit of sandpile groups of rectangular domains with respect to our monomorphisms embeds to the inverse limit of extended sandpile groups which was studied in the previous section. The main problem we put forward is to identify the image of this embedding. For instance, one may expect that it is a rational subtorus or a divisible group (a step in this direction was previously made in \cite{lang2022sandpile}, where it was shown that the limit contains elements of all finite orders). This expectation translates nicely in the dual language, i.e. that the inverse limit of (usual) sandpile groups of rectangular domains under the epimorphisms has no torsion.

It is instructive to see what is the situtation in the one-dimensional case, which is as a general rule dramatically simpler. The extended sandpile group of an integer interval $(0,n)\cap\mathbb{Z}$ for $n\geq 3$ is a two-dimensional torus made of circle-valued arithmetic progressions, and the usual sandpile group is a cyclic group of order $n,$ corresponding to such progressions which take $0\in\mathbb{R}\slash\mathbb{Z}$ at $0$ and $n.$ The direct limit of this groups is $\mathbb{Q}/\mathbb{Z},$ a rational circle corresponding to periodic circle-valued arithmetic progressions with value $0$ at the origin, and the inverse limit under the epimorphisms is $\widehat{\mathbb{Z}}=\prod_{l}\mathbb{Z}_{l},$ a direct product of all additive groups of $l$-adic integers, which indeed has no torsion.

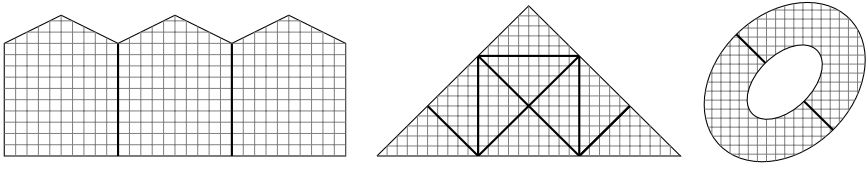
\begin{figure}
    \centering
    \begin{tikzpicture}

        \begin{scope}[scale=1.5]
            \clip[draw](-1,0)--(-1,1)--(-0.5,1.25)--(0,1)--(0.5,1.25)--(1,1)--(1.5,1.25)--(2,1)--(2,0)--cycle;
            \draw[step=0.1,gray,very thin] (-1,0) grid (2,1.5);
            \draw[thick](0,0)--(0,1);
            \draw[thick](1,0)--(1,1);
            
        \end{scope}
       
        \begin{scope}[scale=2,xshift=77]
        \clip [draw] (-1,0)--(0,1)--(1,0)--(-1,0);
        \draw[step=0.06666,very thin,gray] (-1,0) grid (1,1);
        
        \draw[thick] (-0.3333,0)--(-0.3333,0.6666)--(0.3333,0.6666)--(0.3333,0);
        \draw[thick] (-0.3333,0)--(-0.6666,0.3333);
        \draw[thick] (0.3333,0)--(0.6666,0.3333);
        \draw[thick] (0.3333,0)--(0.6666,0.3333);
        \draw[thick] (0,0.3333)--(-0.3333,0);
        \draw[thick] (0,0.3333)--(0.3333,0);
        \draw[thick] (0,0.3333)--(-0.3333,0.6666);
        \draw[thick] (0,0.3333)--(0.3333,0.6666);
        
        \end{scope}

        \begin{scope}[xshift=250,yshift=28,scale=1.2]
            \clip [draw,cm={cos(45) ,-sin(45) ,sin(45) ,cos(45) ,(0 cm,0 cm)}] (0,0) ellipse (0.75 and 1);
            \draw[step=0.1,gray,very thin] (-1,-1) grid (1,1);
            \draw[thick](-1,1)--(1,-1);
            \draw[fill=white,cm={cos(45) ,-sin(45) ,sin(45) ,cos(45) ,(0 cm,0 cm)}] (0,0) ellipse (0.3 and 0.5);   
        \end{scope}
    \end{tikzpicture}
    \caption{Examples of grid foldings of degrees 3, 9 and 2 respectively.}
    \label{fig_folding}
\end{figure}

\section{Discussion}

There is a striking and not fully explored relation \cite{florescu2015sandpiles} between the sandpile model on a rectangular domain and domino tilings on it. The group playing the role in this relation is not the whole sandpile group of the rectangle, but rather a subgroup of elements symmetric with respect to horizontal and vertical reflections. It seems that with a bit of extra work our Theorem \ref{thm_divisibility} for rectangular foldings may be adapted to the case of such subgroups, and therefore give the analogous divisibility result for the numbers of domino tilings.

The category $\DD$ which was passingly introduced in Section 2 imitates by its essence a category of open subsets of some topological space. It makes sense to extend its class of morphisms, leaving the objects intact, by allowing all maps between graphs whose pullbacks on integer-valued functions commute with $\Delta^\circ.$ Morally, this allows for partial coverings, thus we may denote such a category as $\DD^{et}$ with $et$ standing for ``\'etale'', and $\mathcal{H}_A$ is still some kind of a sheaf. It is conceivable that it would be fruitful to study ``the universal extended sandpile group'' defined as the inverse limit (i.e. global sections) of $\mathcal{H}_{\mathbb{R}\slash\mathbb{Z}}$ over this category. In addition, it looks like there are traces of higher homology theory, as the preparations for the proof of Theorem \ref{thm_answer} suggest.

 Studying more general algebraic structures such as sandpile monoid (where non-recurrent states are taken into account) and sandpile magma (where negative values are allowed), as well as their real boundary extensions, is a very promising direction for further exploration. At the moment we don't know if similar functoriality properties to that of an extended sandpile group hold, but some of the most classical entities, such as the first ``sandpile movie'' (which has in part inspired the introduction \cite{lang2019harmonic} of other movies driven by harmonic potentials) suggested by Michael Creutz in \cite{creutz1991abelian} connecting the zero element in the monoid to the neutral element in the group, are naturally dwelling there as geodesic paths. Moreover, it seems to be plausible that the extended sandpile magma is realized geometrically as a tropical modification of the extended sandpile group.

In conclusion, here are a few words about the relevance of tropical view. First of all, another motivation for the  idea of extending the sandpile model has arisen from the following interpretation of the scaling limit theorems stated (in their weaker form) in \cite{kalinin2016tropical}, and with their consequences drawn in \cite{kalinin2018self}: infinitesimal behavior of the sandpile group in high-energy regime near the scaling limit is described by tropical geometry. Since the usual sandpile group is discrete, this sentence is hard to comprehend literally. However, the extension at the boundary, brings the continuous structure, and the extended sandpile group is indeed a tropical Abelian variety. More, this variety may be said to be defined over $\mathbb{Z},$ and its set of integral points is precisely the classical sandpile group. Nevertheless, the corresponding restriction homomorphisms are not defined over integers, and thus there is no functionality under domain inclusions at the level of classical sandpile groups -- this is why the presence of our mono/epi-morphisms in some special situations is so exciting.

\section*{Acknowledgements}
Supported by the Simons Foundation, grant SFI-MPS-T-Institutes-00007697, and the Ministry of Education and Science of the Republic of Bulgaria, grant DO1-239/10.12.2024. The author is grateful to the anonymous referee for the comments which helped to improve the article, and to Andrei Zabolotskii who has pointed out that the first assertion of the original formulation of Theorem \ref{thm_divisibility} was known which motivated a much stronger rephrasing of this theorem not only for rectangles but for arbitrary grid foldings.

\bibliographystyle{tufte}
\bibliography{bibliography.bib}
\end{document}